\documentclass{article}
\usepackage[english]{babel}
\usepackage[affil-it]{authblk}

\usepackage[letterpaper,top=2cm,bottom=2cm,left=3cm,right=3cm,marginparwidth=1.75cm]{geometry}

\usepackage{algpseudocodex}
\usepackage{algorithm}
\algrenewcommand\algorithmicrequire{\textbf{Input:}}
\algrenewcommand\algorithmicensure{\textbf{Output:}}
\usepackage{ifdraft}
\usepackage{amssymb}
\usepackage{rotating}
\usepackage[colorlinks=true, allcolors=blue, breaklinks=true]{hyperref}
\usepackage{subfig} 
\usepackage{appendix}
\usepackage{tipa}
\usepackage{setspace}
\usepackage[absolute]{textpos} 
\usepackage{booktabs,multirow,tabularx} % for better table rules
\usepackage{multicol} % for multiple columns
\usepackage{float}
\usepackage{graphicx,amsmath,amsfonts,amsthm}
\usepackage{supertabular}
\usepackage{makecell}
\usepackage{caption}
\usepackage{tikz}
\usepackage{qtree}
\usepackage{tikz-qtree}
\usepackage{diagbox}

\newtheorem{prop}{Proposition}

\newtheorem{theorem}{Theorem}
\newtheorem{observation}{Observation}
\newtheorem{problem}{Problem}

\begin{document}

\title{Improved lower bounds for the maximum size of \\ Condorcet domains}

\date{}

\author[1,2]{Alexander Karpov}
\author[3]{Klas Markstr{\"o}m} 
\author[4]{S{\o}ren Riis}
\author[5]{Bei Zhou\footnote{Authors are  listed in alphabetical order}}
\affil[1]{HSE University, Moscow, Russia}
\affil[2]{Institute of Control Sciences, Russian Academy of Sciences, Moscow, Russia}
\affil[3]{Ume\aa\ University, Ume\aa\, Sweden} 
\affil[4]{Queen Mary University of London, London, United Kingdom}
\affil[5]{Imperial College London, London, United Kingdom}
\renewcommand\Affilfont{\small}

\maketitle

\begin{abstract}
Condorcet domains are sets of linear orders with the property that, whenever voters' preferences are restricted to the domain, the pairwise majority relation (for an odd number of voters) is transitive and hence a linear order. Determining the maximum size of a Condorcet domain, sometimes under additional constraints, has been a longstanding problem in the mathematical theory of majority voting. The exact maximum is only known for $n\leq 8$ alternatives.

In this paper we use a structural analysis of the largest domains for small $n$ to design a new inductive search method. Using an implementation of this method on a supercomputer, together with existing algorithms, we improve the size of the largest known domains for all $9 \leq n \leq 20$. These domains are then used in a separate construction to obtain the currently largest known domains for $21 \leq n \leq 25$, and to improve the best asymptotic lower bound for the maximum size of a Condorcet domain to $\Omega(2.198139^n)$. Finally, we discuss properties of the domains found and state several open problems and conjectures.
\end{abstract}

%-----------------------------------------------------------
\section{Introduction}
One of the most studied problems in social choice theory is to rank a set of alternatives by applying majority voting to the preferences submitted by a group of voters. Here each voter ranks the alternatives by giving a linear order on that set. Next, a ranking is produced by comparing each pair of alternatives according to the majority vote on their pairwise ranking\footnote{We assume that the number of voters is odd.}. In his seminal work, Condorcet \cite{Con85} showed in 1785 that this procedure does not always lead to even a partial order on the alternatives: preference cycles may occur. A \emph{Condorcet domain} is a set of linear orders such that whenever voters choose their rankings from the domain, majority voting leads to a linear order. This can be thought of as a well-behaved preference landscape.

The first non-trivial model for the possible rankings of a collection of alternatives was Black's single-peaked domain. This domain was introduced formally in \cite{black} and corresponds to a society in which all voters share a common political axis: each voter has a preferred position, the ``peak'', on this axis, and ranks alternatives on each side of the peak by closeness to the peak. Black's single-peaked domain has size $2^{n-1}$, where $n$ is the number of alternatives. Soon after, Arrow \cite{arrow63} introduced what we now call Arrow's single-peaked domains, which generalize Black's domain but have the same maximum size. In fact, as shown in \cite{slinko2019condorcet}, all maximal domains in this class have the same size.

In the 1970s further examples of domains with the same maximum size were found, and Johnson \cite{Johnson78} conjectured that $2^{n-1}$ was the maximum size of a Condorcet domain. This was disproved in \cite{AJ84}, though only by an additive constant. Johnson's conjecture was made again in 1992 by Craven \cite{craven1992social}, who was not aware of the older conjecture and counterexamples. The repeated conjecture stimulated interest, and Fishburn found a new counterexample for $n=5$ \cite{PF:1992}. Later, in \cite{fishburn1996acyclic} he generalized this via his ``alternating scheme'', which produces what are now called Fishburn's domains. These domains have size of order $n\times 2^n$ \cite{GR:2008}, thereby beating the size in the old conjectures by a polynomial factor and, in fact, giving the largest known domains for all small values of $n$ at the time. Fishburn also introduced the replacement scheme and showed that for $n\geq 16$ it gives larger domains than the alternating scheme, and can yield a growth rate $a^n$, where $a>2$. Much later, Raz \cite{Raz} proved an upper bound of the form $\mathcal{O}(b^n)$, for an implicit constant $b>2$, for the maximum size of a Condorcet domain on $n$ alternatives.

Further examples of large domains with various additional conditions were found in \cite{DKK:2012}, for $n=42$ alternatives, and in \cite{karpovslinko} for $n\geq 34$ alternatives. In \cite{Leedham-Green2023condorcet} the largest Condorcet domain for $n=8$ was determined: it has size 224, while the Fishburn domain has size 222.

For $n=9,\ldots,15$ the Fishburn domains remained the largest known domains until \cite{karpovslinko} came up with a new construction combining domains on fewer alternatives, which led to larger domains for $n\geq 13$. For $n=10$ and $n=11$, domains larger than Fishburn's were given in \cite{zhou2025efficient}, found by a machine-learning-guided computer search. In \cite{karpov2025coherent} the class of \emph{coherent domains} was introduced, and it was shown that a method called the odd set-alternating scheme produces domains in this class which are large enough to improve the asymptotic lower bound on the size of maximum Condorcet domains to $2.1973^n$, which was previously the best lower bound.

From a voting-theoretic perspective, one reason for looking for large domains is, as stated for example in \cite{raynaud1982individual}, that large domains correspond to more freedom in individual preferences. However, size is far from the only property of Condorcet domains that has been studied. The survey by Monjardet \cite{monjardet2009acyclic} emphasizes properties with connections to various parts of mathematics. The more recent survey by Puppe and Slinko \cite{puppe2024maximal} covers both properties with a more voting-theoretic emphasis and most of the developments since Monjardet's survey.

The aim of this paper is to take another step in the construction of large Condorcet domains. Based on our study of subdomains of the known maximum domains for $n\leq 8$, we have devised a heuristic search method. Running an implementation of this method on a supercomputer, we have found the largest known Condorcet domains for $n=9,\ldots,20$, and used the method from \cite{karpovslinko} to give improved values for $n=21,\ldots,25$ as well. Using our domains in Fishburn's replacement scheme leads to an improved lower bound of $\Omega(2.198139^n)$ for the maximum size of Condorcet domains. As part of our investigation, we have also determined the maximum size in some subclasses of Condorcet domains for small $n$.

\subsection{Overview}
The structure of this paper is as follows. Section~\ref{sec:background} provides notation and definitions from the theory of Condorcet domains. In Section~\ref{sec:maximum_domains} we introduce the various classes of Condorcet domain which have been studied in the existing literature. Section~\ref{sec:restrictions} discusses restrictions of domains to subsets of the alternatives, and gives a few new enumerative results. In Section~\ref{sec:alg} we present our new algorithm for construction of large Condorcet domains and present the results of our search for $n\leq 20$.  Section~\ref{sec:asymptotic} discusses asymptotic lower bounds and uses the domains from the previous section to give the largest known  asymptotic lower bound for the size of maximum Condorcet domains. Finally, Section~\ref{sec:large_domain_conclusion} discusses the structure of both our new and some previously known large domains and presents several open problems and conjectures.

%-----------------------------------------------------------
\section{Definitions and Notation}
\label{sec:background}

Let $X=[n]=\{1,\dots,n\}$ be the set of alternatives, and let $\mathcal{L}(X)$
denote the set of all linear orders on $X$. Let $N$ be the set of agents, where
each agent $i\in N$ has a preference order $P_i\in \mathcal{L}(X)$. For brevity,
we write a preference order as a string; for example, $12\cdots n$ means that $1$
is the most preferred (ranked first) and $n$ is the least preferred (ranked last).

A subset of linear orders $\mathcal{D}\subseteq \mathcal{L}(X)$ is called a
\emph{domain}. The size of a domain is the number of orders it contains. A domain
$\mathcal{D}$ is a \emph{Condorcet domain} if, whenever the preferences of all
agents belong to $\mathcal{D}$, the majority relation of any preference profile
with an odd number of agents is transitive.

We say that $\mathcal{D}$ is a \emph{maximum} Condorcet domain if it has the
largest size among all Condorcet domains on the same set of alternatives.
A Condorcet domain $\mathcal{D}$ is \emph{maximal} if every Condorcet domain
$\mathcal{D}' \supset \mathcal{D}$ (on the same set of alternatives) coincides
with $\mathcal{D}$.

As usual, a \emph{maximum} domain is a domain of largest possible cardinality,
and a \emph{maximal}  domain is a domain which is not a subset of a larger domain.
Every maximum Condorcet domain is maximal, but not vice versa.  Since our goal is to find maximum domains, unless otherwise specified, we focus on maximal domains. We use the abbreviation MCD for maximal Condorcet domain.

Two domains $\mathcal{D}_1$ and $\mathcal{D}_2$ are \emph{isomorphic} if
$\mathcal{D}_2$ can be obtained from $\mathcal{D}_1$ by relabelling the
alternatives.

Following \cite{akello2024condorcet}, the \emph{core} of a CD is the set of
linear orders in it that map it to itself, and the \emph{dual} of a CD is a CD
obtained by reversing each linear order in it. A CD is \emph{self-dual} if it is
isomorphic to its dual.

A \emph{societal axis} for a domain is a designated linear order on the set of
alternatives. In voting theory this is usually assumed to correspond to some
organising principle for the alternatives, e.g.\ a left--right political scale.
The societal axis is not necessarily an element of the domain, e.g.\ for some
non-maximal domains. Throughout the paper we shall assume that the societal axis is the standard order $12\ldots n$. Note that every Condorcet domain with a societal axis is isomorphic to a domain of this form. 

Sen \cite{Sen1966} proved that a domain is a Condorcet domain if the restriction
of the domain to any triple of alternatives $(a,b,c)$ satisfies a never
condition. A never condition can be of three forms $xNb$, $xNm$, or $xNt$,
referred to as a \emph{never bottom}, \emph{never middle}, and \emph{never top}
condition, respectively. Here $x$ is an alternative from the triple, and $xNb$,
$xNm$, and $xNt$ mean that $x$ is not ranked last, second, or first,
respectively, in the restricted domain. Fishburn noted that for domains with a
societal axis, never conditions can instead be described as $iNj$, $i,j\in[3]$.
Here $iNj$ means that the $i$th alternative from the triple according to the
societal axis does not occupy the $j$th position within this triple in each
order from the domain. For example, the singleton domain $\{abc\}$ on the triple
$a,b,c\in[n]$ with $a<b<c$ satisfies never conditions $1N2,1N3,2N1,2N3,3N1,3N2$,
but violates never conditions $1N1,2N2,3N3$.

We use $\mathcal{N}_{[n]}$ to denote a complete set of triples together with
their assigned never conditions on the set of alternatives $[n]$. We write
$\mathcal{N}_{[n]}(\mathcal{D})$ for a choice of $\mathcal{N}_{[n]}$ that leads
to the CD $\mathcal{D}$, and $\mathcal{D}(\mathcal{N}_{[n]})$ for the CD that
satisfies all the never conditions in $\mathcal{N}_{[n]}$. Moreover,
$\mathcal{N}_{[n]\setminus\{i,j\}}$ denotes the subset of triples in
$\mathcal{N}_{[n]}$ that do not use alternatives $i$ or $j$ (together with their
assigned never conditions), and $\mathcal{U}(\mathcal{N}_{[n]})$ denotes the set
of triples in $\mathcal{N}_{[n]}$ that do not have a never condition assigned.

Fishburn's alternating scheme constructs a set of never conditions by specifying
that if the middle element of a triple is even, it is assigned the never
condition $2N1$, and otherwise $2N3$ \cite{fishburn1996acyclic}. The resulting
domain is commonly called the Fishburn domain \cite{slinko2024family}, as is its
dual domain for odd $n$.

A CD $\mathcal{D}$ is a \emph{full} Condorcet domain if it coincides with the
set of linear orders satisfying the same never conditions as $\mathcal{D}$.

A Condorcet domain $\mathcal{D}$ is \emph{connected} if every pair of orders
from the domain can be connected by a sequence of orders from the domain such
that consecutive orders in the sequence differ by a transposition of neighbouring
alternatives \cite{puppe2024maximal}.

A Condorcet domain $\mathcal{D}$ is a \emph{peak-pit} domain \cite{DKK:2012} if
$\mathcal{D}$ is defined by a complete set of never conditions that contain only
types $iN1$ or $iN3$, $i\in[3]$.

A Condorcet domain $\mathcal{D}$ has \emph{maximum width} if it contains a pair
of orders such that one is the reverse of the other.

Most of the programs used for our work are built on the Condorcet Domain Library
(CDL) \cite{zhou2024cdl}, which provides a wide range of functionalities for
computational problems related to Condorcet domains.

%-----------------------------------------------------------
\section{Maximum domains for different domain types}
\label{sec:maximum_domains}
In the literature on Condorcet domains, the problem of determining maximum domain
sizes has been examined both for general Condorcet domains and for various
subclasses. We list the main such classes and give a few new results on their
sizes.

In \cite{fishburn1996acyclic} Fishburn introduced the size function for
Condorcet domains:
\[
f(n)=\max \{|\mathcal{D}|: \text{$\mathcal{D}$ is a Condorcet domain on a set of $n$ alternatives}\}.
\]
Later, \cite{karpovslinko} introduced size functions for two classes of peak-pit
domains:
\begin{itemize}
    \item $g(n)=\max \{|\mathcal{D}|: \mathcal{D}\text{ is a peak-pit Condorcet domain of maximum width on a set of $n$ alternatives}\}.$
    \item $h(n)=\max \{|\mathcal{D}|: \mathcal{D}\text{ is a peak-pit Condorcet domain on a set of $n$ alternatives}\}.$
\end{itemize}
The class of domains of maximum width, not necessarily peak-pit, is also well
studied, and we denote the maximum size in this class by
\[
g_w(n)=\max \{|\mathcal{D}|: \mathcal{D}\text{ is a Condorcet domain of maximum width on a set of $n$ alternatives}\}.
\]
By definition, $g(n) \leq h(n)\leq f(n)$ and $g(n)\leq g_w(n) \leq f(n)$. It is
not known whether $g_w(n)\leq h(n)$ for all $n$.

\begin{problem}
    Is $g_w(n)\leq h(n)$ for all $n$?
\end{problem}

From the data in \cite{akello2024condorcet} and \cite{Leedham-Green2023condorcet}
we obtain:
\begin{observation}
    For $n\leq 8$, $g(n)=g_w(n)$ and $h(n)=f(n)$.
\end{observation}
It also follows that for $n\leq 8$ the maximum size among domains of maximum
width is achieved only by the Fishburn domains. We have extended these results
slightly by a new computer search.
\begin{observation}
    For $n=9$ we have $g(9)=488$, and the only domains that attain this size are
    the two Fishburn domains.
\end{observation}
\begin{proof}
    This result was obtained by a complete computer search using two independent
    search programs. One program used the algorithm implemented in the CDL
    library \cite{zhou2024cdl}, as explained in
    \cite{zhou2025orderlyalgorithmgenerationcondorcet}. The other program was
    written independently and ran on a different machine, to ensure independence
    in both implementation and hardware.
\end{proof}

It would be desirable to extend this result.
\begin{problem}
    Determine $h(9)$ and $g_w(9)$.
\end{problem}
We attempted to use our general search program on these two problems, but
terminated them after a partial search proved too time-consuming. It is possible
that both are within reach for a more specialised program.

%-----------------------------------------------------------
\section{Restrictions of Condorcet Domains}
\label{sec:restrictions}
The \emph{restriction} of a domain $\mathcal{D}$ to a subset of alternatives
$A\subseteq X$ is the set of linear orders in $\mathcal{L}(A)$ obtained by
restricting each linear order in $\mathcal{D}$ to $A$.

Restrictions to small subsets of alternatives underlie several important
properties of Condorcet domains.

A domain is \emph{ample} if its restriction to any pair of alternatives has size
two. The term \emph{ample} was introduced in the preprint version of
\cite{puppe2024maximal}. Ampleness is a natural property from a voting-theoretic
perspective, since it says that any pair of alternatives can be ranked in both
possible ways. Weaker versions of this property have also been studied in
generalisations of the Gibbard--Satterthwaite theorem \cite{Aswal2003}. The
commonly used property of \emph{minimal richness}, so named in \cite{Aswal2003},
also implies that a domain is ample.

A Condorcet domain $\mathcal{D}$ is \emph{copious} \cite{slinko2019condorcet} if
the restriction of $\mathcal{D}$ to any triple of alternatives has size $4$,
which is the maximum possible size for three alternatives. Copiousness is
important in the study of maximal Condorcet domains, since a copious full domain
is also maximal, while a maximal domain need not be copious. However, it has
been conjectured that all maximal peak-pit domains are copious. A copious domain
is ample, but in \cite{akello2024condorcet} it was shown that there are maximal
domains that are not ample.

The restriction to subsets of five alternatives was used in \cite{zhou2025efficient}
in combination with machine-learning methods to build a heuristic search for
large Condorcet domains. That search found the largest known domains for $n=10$
and $n=11$.

In \cite{karpov2024local} several restriction properties were unified in the
following definition. A domain $\mathcal{D}$ is $(k,s)$-abundant if the
restriction of $\mathcal{D}$ to any subset of $k$ alternatives has size at least
$s$. Thus, ampleness corresponds to $(2,2)$-abundance and copiousness to
$(3,4)$-abundance. The search in \cite{zhou2025efficient} focused on
$(5,16)$-abundant domains, after it had been observed that all large Condorcet
domains for $n\leq 8$ are $(5,16)$-abundant. In \cite{karpov2024local} it was
also shown that for fixed $k$ the maximum possible abundance for sufficiently
large $n$ is $(k,2^{k-1})$.

As noted in the last section, peak-pit domains have received much interest in
the existing literature. The data from \cite{akello2024condorcet} include a
complete enumeration of all peak-pit domains. For $n=8,9$, we have been able to
enumerate the largest domains with sufficiently high abundance, obtaining the
following result.
\begin{observation}
\begin{enumerate}
    \item For $n=5$, all MCDs of size at least $19$ are $(4,8)$-abundant. There
    exist MCDs of size $18$ with lower abundance.

    \item For $n=6$, all MCDs of size at least $42$ are $(4,8)$-abundant. There
    exist peak-pit MCDs of size $41$ with lower abundance. All MCDs of size at
    least $40$ are $(5,16)$-abundant, and there are domains of size $39$ with
    lower abundance.

    \item For $n=7$, all MCDs of size at least $94$ are both $(4,8)$- and
    $(5,16)$-abundant. There are domains of size $93$ for which both abundances
    are lower.

    \item For $n=8$, the numbers of maximal $(4,8)$-abundant peak-pit domains of
    size at least $219$ are
    \[
        \{219: 3, 220: 3, 222: 1, 224: 1\},
    \]
    and the numbers of maximal $(5,16)$-abundant peak-pit domains are
    \[
        \{219: 3+, 220: 63+, 222: 1, 224: 1\}.
    \]
    Here the data are displayed in the format (CD size: number of domains); a
    suffix ``$+$'' indicates that the count is a lower bound.

    \item The maximum size of a $(5,16)$-abundant peak-pit domain for $n=9$ is
    $492$.
\end{enumerate}
\end{observation}
\begin{proof}
The results for $n\leq 7$ were obtained by checking the abundance of the domains
in the complete lists from \cite{akello2024condorcet}. The results for $n=8,9$
were obtained by a computer search using the software from the CDL library.
\end{proof}

Here we note that for $n=8$ there are more maximal domains of size $220$ that are
$(5,16)$-abundant than those that are $(4,8)$-abundant.

%-----------------------------------------------------------
\section{Inductive Domain Construction}
\label{sec:alg}
The abundance properties discussed so far, and the variants used in the search
method of \cite{zhou2025efficient}, all consider restrictions of domains to
subsets of a fixed size. However, using the existing data from
\cite{akello2024condorcet,Leedham-Green2023condorcet}, we can also make the
following observation.
\begin{observation}\label{obs:inductive-structure}
\begin{enumerate}
    \item For $4\leq n\leq 7$, each maximum domain has a restriction to $n-1$
    alternatives that is itself a maximum domain.

    \item The maximum domain for $n=8$ has a restriction to $7$ alternatives of
    size $96$.

    \item In each case there exist two alternatives $i,j$ such that the
    restrictions to $[n]\setminus\{i\}$ and $[n]\setminus\{j\}$ are isomorphic
    domains of the sizes given above, and they have identical restrictions to
    $[n]\setminus\{i,j\}$.
\end{enumerate}
\end{observation}

For MCDs in general, we have very little control over how much the size of a
domain can drop when we restrict from $n$ alternatives to $n-1$ alternatives.
However, Observation~\ref{obs:inductive-structure} suggests that a \emph{maximum}
domain typically has a restriction that remains close to maximum size, and that
it has two large isomorphic $n\!-\!1$ restrictions with a large common overlap.

Based on this, we construct an inductive search algorithm that builds domains on
$n+1$ alternatives from a set of large domains on $n$ alternatives. As we will
see, this allows us to find the largest known Condorcet domains up to $n=20$ by
direct search, and up to $n=25$ via the theoretical construction method using
these domains as input.

%--------------------------------------------------------------------

\subsection{The inductive extension step}
\label{sec:inductive-construction}

\paragraph{From large restrictions to a small completion problem.}
Fix $n\ge 4$ and suppose we want to construct a large Condorcet domain on
$[n{+}1]$. The key idea is to \emph{enforce} the restriction pattern from
Observation~\ref{obs:inductive-structure}.  Choose two distinct alternatives
$i,j\in[n{+}1]$ and write
\[
A=[n{+}1]\setminus\{i\},\qquad B=[n{+}1]\setminus\{j\},\qquad
C=[n{+}1]\setminus\{i,j\}.
\]
We embed large domains on $[n]$ as domains on $A$ and on $B$ (via the
order-preserving relabelling). If the two embedded domains induce the \emph{same}
never conditions on the overlap $C$ (equivalently, their restrictions to $C$
coincide), then we may take the union of their never conditions. This union
already assigns never conditions to every triple that avoids $i$ (from the domain
on $A$) and to every triple that avoids $j$ (from the domain on $B$). Hence the
only triples left undecided are precisely the $n-1$ triples containing both $i$
and $j$, namely $\{i,j,k\}$ for $k\in C$. The extension step therefore reduces
to a constrained \emph{completion} problem on these remaining $n-1$ triples.

\paragraph{Example.}
For $i=n{+}1$ and $j=1$, one embedded domain lives on
$[n]=[n{+}1]\setminus\{n{+}1\}$ and the other on
$\{2,\ldots,n{+}1\}=[n{+}1]\setminus\{1\}$. Compatibility means that their
restrictions to $\{2,\ldots,n\}$ coincide, and the only unassigned triples are
$\{1,n{+}1,k\}$ for $k=2,\ldots,n$.

\paragraph{Why this is not automatic.}
This gluing step is delicate. A large CD typically has many labelled
realisations (relabelings of alternatives), and different labelings can lead to
very different overlaps on $C$. Moreover, even when two instantiations agree on
$C$, the resulting partial set of never conditions can constrain the remaining
$\{i,j,k\}$-triples so strongly that most completions produce domains that are
not particularly large. Thus, compatibility on the overlap is necessary but in
practice far from sufficient: the search must explore many potential overlaps
and prune aggressively during completion.

\paragraph{Empirical pruning and a ``greedy'' phenomenon.}
Empirical tests for $n\leq 11$ revealed two regularities among the largest
domains found: they were peak-pit, and they were copious. Accordingly, in the
completion step we restrict to peak-pit never conditions and discard any
non-copious partial assignment.  More surprisingly, experiments also indicated a
robust greedy strategy: starting from the largest domain(s) found on $n$
alternatives, enumerating \emph{all} compatible overlaps (in the sense above) and
then searching completions of the remaining $n-1$ triples repeatedly produced
the next record domain (up to relabelling) for every $n$ tested ($8\le n\le 20$).
This observation motivates the deterministic shortcut described below.

\subsubsection{The algorithm}
We are now ready to describe our main algorithm. In Appendix~A, we give a pseudocode version of this description. Let $L_n=\{\mathcal{C}_1,\dots,\mathcal{C}_m\}$ be a list of (non-isomorphic) MCDs on $n$ alternatives. To generate candidates on $[n{+}1]$ alternatives we proceed in three steps.

\paragraph{Step 1 (instantiations on $n$-subsets).}
For each $i\in[n{+}1]$ and each $\mathcal{C}_r\in L_n$ we form the instantiated
domain
\[
  \mathcal{C}_r^{\setminus i}:=\mathcal{C}_r\bigl([n{+}1]\setminus\{i\}\bigr).
\]
Here $\mathcal{C}(A)$ denotes the domain on $A$ obtained from $\mathcal{C}$ by
relabeling alternatives via the unique order-preserving bijection $[n]\to A$.

\paragraph{Step 2 (enumerating compatible overlaps).}
Fix an ordered pair $(i,j)$ with $i\ne j$. Two instantiations
$\mathcal{C}_r^{\setminus i}$ and $\mathcal{C}_s^{\setminus j}$ are
\emph{compatible} if they assign the same never condition to every triple
from the intersection of their sets of alternatives,  $[n{+}1]\setminus\{i,j\}$. When compatibility holds, we
take the union of their sets of never conditions; otherwise the pair is
discarded. The resulting incomplete sets of never conditions are collected in an
intermediate list $L_{n+1,\mathrm{partial}}$.

\smallskip
\noindent
\emph{Implementation note.} To enumerate all compatible overlaps efficiently, we
group instantiations by their restriction to $[n{+}1]\setminus\{i,j\}$ (i.e.\ by
the subset of never conditions on triples avoiding $\{i,j\}$), and only merge
pairs within the same group.

\paragraph{Step 3 (completion with pruning).}
Each remaining partial set of never conditions is completed by a backtrack 
search that assigns peak-pit never conditions to the unassigned triples
$\{i,j,k\}$. Non-copious partial solutions are discarded during the search. We
compute the size of each completed domain and retain only those above a preset
size threshold. Finally, we reduce the resulting list by removing isomorphic
copies to obtain the next list $L_{n+1}$.

\paragraph{A deterministic shortcut for $n\ge 8$.}
Empirical runs revealed a faster path to the largest domains. Let
$\mathcal{C}^{\mathrm{max}}_n$ be the largest CD found on $n$ alternatives.
Define $i=\lfloor n/2\rfloor$ and $j=i+1$. Repeating Steps~1--3 for this single
pair $(i,j)$ alone suffices to reconstruct the largest domain on $n{+}1$
alternatives for every $n$ tested ($8\le n\le 20$).

%--------------------------------------------------------------------

%---------------------------------------
\section{Results}
Using the algorithm from the previous section, we performed a search for
$8\leq n \leq 20$. As input, we started from the domains of size at least $90$
for $n=7$. At each step we generated candidates on $n+1$ alternatives and kept
those whose sizes were close to the largest size found. For $n\geq 16$ we kept
only the single largest domain. As $n$ increases, the search quickly becomes
both time-consuming and memory-intensive, so it was run on two different Linux
clusters using a few thousand CPU cores.

In Table~\ref{tab:inductive} we list the largest domain sizes found using this
method, together with the previously best lower bounds. The $n=8$ CD of size
$224$ was found in \cite{Leedham-Green2023condorcet}. For $n=9$ and $n=12$ the
previous lower bounds are given by the Fishburn domains. The previous lower
bounds for $n=10$ and $n=11$ are from \cite{zhou2025efficient}, and the largest
domains previously known for $13\leq n \leq 20$ were constructed in
\cite{karpovslinko}.

As the table shows, our search improves the existing lower bounds for
$9\leq n\leq 20$. The same approach would likely improve the bounds for larger
$n$ as well, but we stopped at $n=20$ since the search was already taking
multiple days on several thousand cores.

\begin{table}[H]
\centering
\captionsetup{width=0.9\linewidth}
\caption{Largest CDs found via the inductive search. The bold number indicates it is an improved lower bound over the old one. }
\label{tab:inductive}
\begin{tabular}{cccc}
\toprule
\textbf{n} & \textbf{Fishburn domains} & \textbf{Previous lower bounds} & \textbf{Improved lower bounds} \\ 
\midrule
8 & 222 & 224 & 224 \\
9 & 488  & 488 & \textbf{492} \\
10 & 1069 & 1082 & \textbf{1104} \\
11 & 2324 & 2349 & \textbf{2452} \\
12 & 5034 & 5034 & \textbf{5520} \\
13 & 10840 & 10940 & \textbf{12182} \\
14 & 23266 & 24481 & \textbf{27026} \\
15 & 49704 &  54752 & \textbf{59430} \\
16 & 105884 & 123004 & \textbf{131536} \\
17 & 224720 & 271758  & \textbf{292015} \\
18 & 475773 & 602299 & \textbf{649079} \\
19 & 1004212 &  1323862  & \textbf{1432836} \\
20 & 2115186 &  2917604  & \textbf{3155366} \\
\bottomrule
\end{tabular}
\end{table}

%-----------------------------
\subsection{Properties of the new domains}
While our focus has been to find the largest known Condorcet domains for each
$n$, there are many other properties of interest, both combinatorially and from
a voting-theoretic perspective.

We have checked the abundance properties of the largest domains for $n\geq 8$.
As for smaller $n$, these domains meet (and for $n=8$ exceed) the asymptotically
maximal abundance levels:
\begin{observation}
    The maximum domain for $n=8$ is $(4,8)$- and $(5,18)$-abundant. For $n\geq 9$
    the largest domains found are exactly $(4,8)$- and $(5,16)$-abundant.
\end{observation}

In Table~\ref{tab:properties} we summarise several such properties. The first two
columns give $n$ and the number of non-isomorphic largest domains found for that
value of $n$. All the largest domains we have found are peak-pit domains, even
though for $8\leq n\leq 11$ our search allowed all possible never conditions.

The column labelled \textbf{2N1+2N3} comes from an analysis of the never
conditions satisfied by the domains. Our search focuses on peak-pit never
conditions of the forms $xN1$ and $xN3$. However, among the isomorphic copies of
a given domain, the numbers of never conditions of the four types $1N3$, $3N1$,
$2N1$, and $2N3$ may vary. For each domain we found an isomorphic copy with very
few never conditions of types $2N1$ and $2N3$; in this column we report the
minimum value of $2N1+2N3$ over all isomorphic copies. Since the total number of
never conditions is ${n \choose 3}$, this means that the defining never
conditions are almost exclusively of the forms $1N3$ and $3N1$. This observation
is of particular interest due to its connection to coherent domains, which
satisfy only $1N3$ and $3N1$ never conditions. As mentioned in the introduction,
a special class of coherent domains yields the best known asymptotic lower bound
for the maximum size of Condorcet domains, and we return to this class in the
last section.

The \emph{core} size \cite{akello2024condorcet} is the number of linear orders in
the domain that, when viewed as permutations, map the domain to itself. This is
one measure of the symmetry of the domain: as Table~\ref{tab:properties} shows,
the domains up to $n=13$ have a non-trivial core, while for $n\geq 14$ the core
is trivial. This indicates a change in the structure of the domains for
$n\geq 14$, and possibly that we are no longer finding the actual largest
domains.

One well-studied property of voting domains is the number of alternatives that
can be ranked first or last in an order from the domain. A domain in which every
alternative can be ranked first is called \emph{minimally rich}, and one in which
only one alternative can be ranked first is called \emph{dictatorial}; see
\cite{Aswal2003} for a fuller discussion. For some Condorcet-domain types, such
as Arrow's single-peaked domains, this is a well understood property
\cite{slinko2019condorcet,markstrom2024arrow}, but for general Condorcet domains
there are few results. In the column \textbf{(\#First, \#Last)} we list the
numbers of distinct alternatives that appear first and last, respectively,
across all linear orders in the domain. The domains are far from being minimally
rich, but they move further away from being dictatorial as $n$ increases.

The next two columns indicate whether the domain is self-dual and whether it is
connected. Since these domains are peak-pit domains, they are expected to be
connected \cite{puppe2024maximal}. The final column gives the diameter of each
domain in the adjacent-transposition graph: the maximum, over all pairs of
orders $x$ and $y$ in the domain, of the length of a shortest sequence of
adjacent transpositions that connects $x$ to $y$ through orders in the domain. A
Condorcet domain of maximum width has diameter ${n \choose 2}$, realised by a
pair of reverse orders \cite{monjardet2009acyclic}, but for our domains the
diameter appears to grow much more slowly. An empirical fit to the data agrees
well with growth on the order of $\mathcal{O}(n(\log n)^{2.5})$, but we do not
yet understand the structure of these domains well enough to prove such a bound.

\begin{table}[H]
\centering
\captionsetup{width=1\linewidth}
\caption{Properties of the new domains. For a CD, the 2N1+2N3 values are the least number of never conditions that are not 1N3 or 3N1 in all its isomorphic forms. (\#First, \#Last) are the number of distinct alternatives ranked as the first, and the last in all linear orders of a CD. }
\label{tab:properties}
\begin{tabular}{cccccccc}
\toprule
\textbf{n} & \textbf{Count} & \textbf{2N1+2N3} & \textbf{Core size} & \textbf{(\#First, \#Last)} & \textbf{Self-dual} & \textbf{Connected} & \textbf{Diameter} \\ 
\midrule
    8 &  1 & 2 & 4 & (4, 4) & Yes & Yes & 20 \\
    9 &  2 & 2 & 2 & (4, 4) & No & Yes & 25 \\
    10 & 2 & 3 & 4 & (4, 4) & No & Yes & 31 \\
    11 & 2 & 3 & 2 & (4, 4) & No & Yes & 37 \\
    12 & 1 & 4 & 4 & (4, 4) & Yes & Yes & 44 \\
    13 & 2 & 6 & 2 & (5, 4), (4, 5) & No & Yes & 52 \\
    14 & 1 & 8 & 1 & (5, 5) & Yes& Yes & 61 \\
    15 & 2 & 9 & 1 & (5, 5) & No & Yes & 69  \\
    16 & 2 & 9 & 1 & (5, 5) & No & Yes & 78 \\
    17 & 2 & 10 & 1 & (5, 5) & No  & Yes &  87  \\
    18 & 2 & 12 & 1 & (5, 5) & No &   Yes  & 98   \\
    19 & 2 & 12 & 1 & (5, 5) & No  &   Yes  &   106 \\
    20 & 2  & 16  & 1 & (6, 5), (5,6)  & No  &     &    \\
\bottomrule
\end{tabular}
\end{table}

%-----------------------------------------------------
\subsection{Poset Structure for largest subdomains}
\label{sec:subCDforest}
For $n\geq 9$, the largest domains we have found all have the property that they
can be constructed from a largest domain on $n-1$ alternatives. As a result, for
each such domain and for any $k\geq 8$, there exist at least two subsets of the
alternatives of size $k$ such that the restriction to that subset is a largest
domain on $k$ alternatives. This means that the collection of subsets of
alternatives that induce a largest domain (for their size) forms a subposet of
the Boolean lattice. As Table~\ref{tab:properties} shows, for each $n$ there is
either a single largest domain, or a pair of dual domains. Since duality does
not change the poset at hand, we only have one restriction poset for each $n$ in
this range.

In Table~\ref{tab:subCDcounts} we show, for each $n$, the number of subsets of
size $k$ whose restriction gives a largest subdomain on $k$ alternatives. As the
table shows, these counts are typically even numbers, reflecting the fact that
there are always at least two distinct restrictions to $n-1$ alternatives.
However, for $n=13$ we find $7$ restrictions to $11$ alternatives. As for the
core in the previous section, this again singles out $n=13$ as an exceptional
case.

Even though we may reach the same subset via different chains of consecutive
restrictions---showing that our partial order is not a tree---we still observe an
apparently exponential growth in the total number of largest subdomains on a
given number of alternatives.

\begin{table}[ht]
\centering
\caption{Number of record-sized sub-CDs (up to isomorphism) of each smaller order 
found inside the largest domain(s) for $n=9,\dots,16$. 
A dash (--) indicates that the given sub-CD order does not arise or is not applicable 
for that $n$. }
\label{tab:subCDcounts}
\scalebox{0.85}{
\begin{tabular}{ccccccccccccc}
\toprule
\textbf{$n$} & 
\textbf{\# 8} & 
\textbf{\# 9} & 
\textbf{\# 10} & 
\textbf{\# 11} & 
\textbf{\# 12} & 
\textbf{\# 13} & 
\textbf{\# 14} & 
\textbf{\# 15} &
\textbf{\# 16} &
\textbf{\# 17} &
\textbf{\# 18} &
\textbf{\# 19}  \\
\midrule
9  & 2   & --  & --  & --  & --  & --  & --  & --  & -- & -- & -- & --  \\
10 & 4   & 2   & --  & --  & --  & --  & --  & --  & -- & -- & -- & --  \\
11 & 8   & 8   & 2   & --  & --  & --  & --  & --  & -- & -- & -- & --  \\
12 & 16  & 16  & 8   & 4   & --  & --  & --  & --  & -- & -- & -- & --  \\
13 & 32  & 28  & 16  & 7   & 2   & --  & --  & --  & -- & -- & -- & --  \\
14 & 64  & 48  & 32  & 12  & 4   & 4   & --  & --  & -- & -- & -- & --  \\
15 & 96  & 92  & 56  & 26  & 8   & 8   & 2   & --  & -- & -- & -- & --  \\
16 & 120 & 149 & 94  & 50  & 16  & 16  & 4   & 2   & -- & -- & -- & --  \\
17 & 160 & 244 & 168  &  94 & 32  & 32  & 8 & 4 & 2 & -- & -- & --  \\
18 & 224 & 348 & 248 & 138  & 48  & 56  & 16 & 8 & 4 & 2 & -- & --  \\
19 & 256 & 480 & 352  & 204 &  72 & 84  & 24 & 12 & 12 & 4 & 2 & --  \\
20 & 384 & 632 & 528  &  284 & 108  & 162  & 54 & 24 & 24 & 8 & 4 & 2  \\
\bottomrule
\end{tabular}}
\end{table}

%-----------------------------------------------------
\section{New Asymptotic lower bounds}
\label{sec:asymptotic}
While the previous sections have focused on finding large domains for fixed $n$,
we now discuss how such domains can be used to obtain asymptotic lower bounds on
domain sizes.

We describe two constructions that produce larger domains from domains on disjoint
sets of alternatives, and hence yield sequences of domains (for increasing numbers
of alternatives) in the classes corresponding to the size functions defined above.
\begin{enumerate}
    \item Given two domains $\mathcal{D}_1$ and $\mathcal{D}_2$ on disjoint sets
    of alternatives $A$ and $B$, and two partitions $A=A_1\cup A_2$ and
    $B=B_1\cup B_2$, the \emph{$1N3$--$3N1$ construction} produces a domain
    $S_1(\mathcal{D}_1,\mathcal{D}_2,A_1,A_2,B_1,B_2)$ on $A\cup B$ with the
    following never conditions on triples that meet both $A$ and $B$:
    \begin{itemize}
        \item For triples $\{a_1,a_2,b\}$ with $a_1,a_2\in A$ and $b\in B$, assign
        $1N3$ if $\{a_1,a_2\}\subseteq A_1$, and otherwise assign $3N1$.
        \item For triples $\{a,b_1,b_2\}$ with $a\in A$ and $b_1,b_2\in B$, assign
        $3N1$ if $\{b_1,b_2\}\subseteq B_1$, and otherwise assign $1N3$.
    \end{itemize}

    \item The second construction produces a domain $S_2(\mathcal{D}_1,\mathcal{D}_2)$
    on $A\cup B$ by assigning $2N3$ to every triple that meets both $A$ and $B$,
    while all other triples keep the never condition they had in $\mathcal{D}_1$
    or $\mathcal{D}_2$.
\end{enumerate}

\begin{prop}\label{superadde}
    Let $e\in \{f,g,h\}$ be one of the functions defined above. Then
    \begin{equation}\label{superadd}
        e(n+m) \geq 2e(n)e(m).
    \end{equation}
\end{prop}

\begin{proof}
Let $\mathcal{D}_1$ be a Condorcet domain on $n$ alternatives of size $e(n)$, and
let $\mathcal{D}_2$ be a Condorcet domain on $m$ alternatives of size $e(m)$,
both chosen to be largest in the relevant class. Under the $1N3$--$3N1$
construction, every concatenation $uv$ with $u\in \mathcal{D}_1$ and
$v\in \mathcal{D}_2$ belongs to a Condorcet domain on $n+m$ alternatives. In
addition, writing $u=u_1\cdots u_n$ and $v=v_1\cdots v_m$, the concatenation
\[
u_1\cdots u_{n-1}\, v_1\, u_n\, v_2\cdots v_m
\]
also belongs to the constructed domain. Thus the constructed domain has size at
least $2e(n)e(m)$. If $\mathcal{D}_1$ and $\mathcal{D}_2$ are peak-pit domains,
then the resulting domain is also peak-pit.

For domains of maximum width, the second construction can instead be used: in
addition to all concatenations $uv$ we also include all concatenations $vu$.
Since both input domains contain a pair of reverse orders, the resulting domain
again has maximum width.
\end{proof}

For the second construction we can also guarantee that, for fixed $k$, abundance
does not become asymptotically suboptimal.
\begin{prop}
    If $\mathcal{D}_1$ and $\mathcal{D}_2$ are $(k,2^{k-1})$-abundant for each
    $k\leq t$, then so is $S_2(\mathcal{D}_1,\mathcal{D}_2)$.
\end{prop}
\begin{proof}
This is immediate for $k$-tuples that contain alternatives from only one of $A$
and $B$. Otherwise, let the $k$-tuple contain a subset $T_1\subseteq A$ and a
subset $T_2\subseteq B$, with $|T_1|+|T_2|=k$. Then
$S_2(\mathcal{D}_1,\mathcal{D}_2)$ contains every concatenation of a linear
order from the restriction of $\mathcal{D}_1$ to $T_1$ (of which there are
$2^{|T_1|-1}$) with a linear order from the restriction of $\mathcal{D}_2$ to
$T_2$ (of which there are $2^{|T_2|-1}$), and also the same number of
concatenations in the opposite order. This gives a total of
\[
2\cdot 2^{|T_1|-1}\cdot 2^{|T_2|-1}=2^{k-1}
\]
linear orders in the restriction to this $k$-tuple.
\end{proof}

Proposition~\ref{superadde} also yields an asymptotic lower bound on the growth
rate.
\begin{prop}\label{1/n}
For any fixed integer $k\ge 1$ we have
\[
\liminf_{n\rightarrow \infty} e(n)^{1/n}\geq (2e(k))^{1/k},
\]
where $e\in \{f,h\}$.
\end{prop}

\begin{proof}
By \eqref{superadd}, for any positive integer $a$ we have
$e(ak)\ge 2^{a-1}e(k)^a$. Writing $n=ak+b$ with $0\le b<k$ and using
\eqref{superadd} once more gives
\[
e(n)=e(ak+b)\ge 2\,e(ak)e(b)\ge 2^{a}\,e(k)^a\,e(b)
= (2e(k))^{(n-b)/k}\,e(b).
\]
Taking $n$th roots and letting $n\to\infty$ yields
$\liminf_{n\to\infty} e(n)^{1/n}\ge (2e(k))^{1/k}$.
\end{proof}

Using our data, we obtain that the growth rate for peak-pit domains (and hence
for general Condorcet domains) is at least $2.18902$. This is lower than the
best bound $2.1973$ obtained in \cite{karpov2025coherent} for peak-pit domains,
despite the fact that our search yields larger domains for $n\leq 20$ than the
method from \cite{karpov2025coherent}.

Instead of requiring domains of a specified type, we can consider the growth
rate for $f(n)$ itself:
\begin{theorem}
    $\lim_{n\rightarrow\infty} f(n)^{1/n}\geq 2.198139$.
\end{theorem}
\begin{proof}
Fishburn \cite{FISHBURN1996209} showed that $f(n+m-1)\geq f(n)f(m)$, and as a
consequence $\lim_{n\rightarrow\infty} f(n)^{1/n}\geq f(k)^{1/(k-1)}$ for every
fixed $k$. Using our constructed domains, for $k=20$ we obtain the lower bound
$3155366^{1/19}=2.198139495\ldots$.
\end{proof}
This gives the currently best lower bound for the asymptotic growth rate of
maximum Condorcet domains.

%-----------------------------------------------------------
% \section{Discussion}
% \label{sec:large_domain_conclusion}
% In Table \ref{tab:largest_sizes} we present our new maximum sizes and lower bounds, together with some growth rate comparisons.

% {\bf This section will be rewritten with added data and discussion in the second upload of this manuscript.}

%-----------------------------------------------------------
\section{Summary of bounds}
\label{sec:large_domain_conclusion}

% This is a preliminary version which will be updated with additional material

Table~\ref{tab:largest_sizes} summarises the best currently known lower bounds
for the maximum size of Condorcet domains.  For $21\le n\le 25$, the entries
marked with subscript $_{ks}$ are obtained by applying the construction of~\cite{karpovslinko}
to the domains found for smaller $n$.

\clearpage

\begin{table}[ht]
\centering
\caption{Maximum size of Condorcet domains. Our improved values are stated in boldface.  Numbers underscored with $ks$ use our domains for lower $n$ in the construction from \cite{karpovslinko}.  }
\label{tab:largest_sizes}
\begin{tabular}{cccccccc}
\toprule
\textbf{n} & \textbf{Fishburn} & $f(n)$ \& h(n)&
\textbf{$x_n/x_{n-1}$}&
\textbf{$\sqrt{x_n/x_{n-2}}$}&
$\lceil\frac{4}{5^{1.5}}\sqrt{5}^n \rceil$ &$ g(n)$ & \\ 
\midrule
3 & 4   & 4                     & 2     & 2     & 4 & 4 &  \\
4 & 9   & 9                     & 2.25  & 2.121 & 9 & 9 &  \\
5 & 20  & 20                    & 2.222 & $\sqrt{5}=2.2360$ & 20& 20&  \\
6 & 45  & 45                    & 2.250 & $\sqrt{5}$& 45& 45 &   \\
7 & 100 & 100                   & 2.222 & $\sqrt{5}$ & 100& 100 &  \\
8 & 222   & 224                 & 2.240 & 2.231 & 224& 222&  \\
9 & 488   & $\geq$\textbf{492}  & 2.196 & 2.218 & 500& 488 & \\
10 & 1069 & $\geq$\textbf{1104} & 2.244 & 2.220 & 1119& - &  \\
11 & 2324 & $\geq$\textbf{2452} & 2.221 & 2.232 & 2500& - &  \\
12 & 5034 & $\geq$\textbf{5520} & 2.251 & 2.236 & 5591& - &  \\
13 & 10840 & $\geq$\textbf{12182} & 2.206 & 2.228& 12500& - &  \\
14 & 23266 & $\geq$\textbf{27026} & 2.218 & 2.212& 27951& - &  \\
15 & 49704 & $\geq$\textbf{59430} & 2.198 & 2.208 & 62500 & - &  \\
16 & 105884& $\geq$\textbf{131536} & 2.213 & 2.206& 139755& - & \\
17 & 224720& $\geq$\textbf{292015} & 2.220 & 2.216& 312500& - &  \\
18 & 475773& $\geq$\textbf{649079} & 2.222 &2.221 & 698772& - &  \\
19 &1004212& $\geq$\textbf{1432836} &2.207 &2.215& 1562500& - &   \\
20 &2115186& $\geq$\textbf{3155366} &2.202 &2.204& 3493857& - &  \\
\hline
21 &4443896& $\geq$\textbf{6674931}$_{ks}$ &2.115&2.221& 7812500 & - &  - \\
22 &9319702& $\geq$\textbf{15037178}$_{ks}$  &2.253&2.183& 17469282 & - & -  \\
23 &19503224& $\geq$\textbf{33162136}$_{ks}$ &2.205&2.229& 39062500 & - &  - \\
24 &40750884& $\geq$\textbf{73279872}$_{ks}$ &2.209&2.207& 87346406 & - &  - \\
25 &84990640&    $\geq$\textbf{160316046}$_{ks}$ & 2.187 & 2.198 & 195312500 &  & \\
\bottomrule
\end{tabular}
\end{table}

%-----------------------------------------------------------
\section*{Acknowledgement}
The Basic Research Program of the National Research University Higher School of Economics partially supported Alexander Karpov. This research utilised Queen Mary's Apocrita HPC facility, supported by QMUL Research-IT. This research was conducted using the resources of High Performance Computing Center North (HPC2N). 

\clearpage
%-----------------------------------------------------------
%\bibliographystyle{plain}
%\bibliography{sample}

\begin{thebibliography}{10}

\bibitem{AJ84}
J.M. Abello and C.R. Johnson.
\newblock How large are transitive simple majority domains?
\newblock {\em SIAM Journal on Algebraic Discrete Methods}, 5(4):603--618,
  1984.

\bibitem{akello2024condorcet}
Dolica Akello-Egwel, Charles Leedham-Green, Alastair Litterick, Klas
  Markstr{\"o}m, and S{\o}ren Riis.
\newblock Condorcet domains on at most seven alternatives.
\newblock {\em Mathematical Social Sciences}, 133:23--33, 2024.

\bibitem{arrow63}
K.J. Arrow.
\newblock {\em Social choice and individual values}, volume~12.
\newblock Yale university press, second edition, 1963.

\bibitem{Aswal2003}
Navin Aswal, Shurojit Chatterji, and Arunava Sen.
\newblock Dictatorial domains.
\newblock {\em Economic Theory}, 22(1):45--62, 2003.

\bibitem{black}
D.~Black.
\newblock On the rationale of group decision-making.
\newblock {\em Journal of Political Economy}, 56(1):23--34, 1948.

\bibitem{craven1992social}
John Craven.
\newblock {\em Social choice: a framework for collective decisions and   individual judgements}.
\newblock Cambridge University Press, 1992.

\bibitem{DKK:2012}
V.I. Danilov, A.V. Karzanov, and G.A. Koshevoy.
\newblock Condorcet domains of tiling type.
\newblock {\em Discrete Applied Mathematics}, 160(7-8):933--940, 2012.

\bibitem{Con85}
Le~Marquis de~Condorcet.
\newblock {\em Essai sur l'application de l'analyse \`{a} la probabilit\'{e}
  des d\'{e}cisions rendues \`{a} la pluralit\'{e} des voix}.
\newblock L'Imprimerie Royale, Paris, 1785.

\bibitem{PF:1992}
P.C. Fishburn.
\newblock Notes on {C}raven's conjecture.
\newblock {\em Social Choice and Welfare}, 9:259--262, 1992.

\bibitem{fishburn1996acyclic}
Peter Fishburn.
\newblock Acyclic sets of linear orders.
\newblock {\em Social choice and Welfare}, 14(1):113--124, 1996.

\bibitem{FISHBURN1996209}
Peter~C. Fishburn.
\newblock Decision theory and discrete mathematics.
\newblock {\em Discrete Applied Mathematics}, 68(3):209--221, 1996.

\bibitem{GR:2008}
A.~Galambos and V.~Reiner.
\newblock Acyclic sets of linear orders via the {B}ruhat orders.
\newblock {\em Social Choice and Welfare}, 30(2):245--264, 2008.

\bibitem{Johnson78}
R.C. Johnson.
\newblock Remarks on mathematical social choice.
\newblock Technical report, Dept. Economics and Institute for Physical Science
  and Technology, Univ. Maryland, 1978.

\bibitem{karpov2024local}
Alexander Karpov, Klas Markstr{\"o}m, S{\o}ren Riis, and Bei Zhou.
\newblock Local diversity of {C}ondorcet domains.
\newblock {\em arXiv preprint arXiv:2401.11912}, 2024.

\bibitem{karpov2025coherent}
Alexander Karpov, Klas Markstr{\"o}m, S{\o}ren Riis, and Bei Zhou.
\newblock Coherent domains and improved lower bounds for the maximum size of  {C}ondorcet domains.
\newblock {\em Discrete Applied Mathematics}, 370:57--70, 2025.

\bibitem{karpovslinko}
Alexander Karpov and Arkadii Slinko.
\newblock Constructing large peak-pit {C}ondorcet domains.
\newblock {\em Theory and Decision}, 94(1):97--120, 2023.

\bibitem{Leedham-Green2023condorcet}
Charles Leedham-Green, , Klas Markstr{\"o}m, and S{\o}ren Riis.
\newblock The largest {C}ondorcet domains on 8 alternatives.
\newblock {\em Social Choice and Welfare}, 62:109--116, 2024.

\bibitem{markstrom2024arrow}
Klas Markstr{\"o}m, S{\o}ren Riis, and Bei Zhou.
\newblock Arrow's single peaked domains, richness, and domains for plurality
  and the {B}orda count.
\newblock {\em arXiv preprint arXiv:2401.12547}, 2024.

\bibitem{monjardet2009acyclic}
Bernard Monjardet.
\newblock Acyclic domains of linear orders: a survey.
\newblock {\em The Mathematics of Preference, Choice and Order: Essays in Honor
  of Peter C. Fishburn}, pages 139--160, 2009.

\bibitem{puppe2024maximal}
Clemens Puppe and Arkadii Slinko.
\newblock Maximal {C}ondorcet domains. a further progress report.
\newblock {\em Games and Economic Behavior}, 145:426--450, 2024.

\bibitem{raynaud1982individual}
Herv{\'e} Raynaud.
\newblock {\em The individual freedom allowed by the value restriction  condition}.
\newblock Institute for Mathematical Studies in the Social Sciences, 1982.

\bibitem{Raz}
Ran Raz.
\newblock V{C}-dimension of sets of permutations.
\newblock {\em Combinatorica}, 20(1):1--15, 2000.

\bibitem{Sen1966}
A.~K. Sen.
\newblock A {P}ossibility theorem on majority decisions.
\newblock {\em Econometrica}, 34(2):491--499, 1966.

\bibitem{slinko2019condorcet}
Arkadii Slinko.
\newblock Condorcet domains satisfying {A}rrow's single-peakedness.
\newblock {\em Journal of Mathematical Economics}, 84:166--175, 2019.

\bibitem{slinko2024family}
Arkadii Slinko.
\newblock A family of {C}ondorcet domains that are single-peaked on a circle.
\newblock {\em Social Choice and Welfare}, 63:57--67, 2024.

\bibitem{zhou2024cdl}
Bei Zhou, Klas Markstr{\"o}m, and S{\o}ren Riis.
\newblock Cdl: A fast and flexible library for the study of permutation sets
  with structural restrictions.
\newblock {\em SoftwareX}, 28:101951, 2024.

\bibitem{zhou2025orderlyalgorithmgenerationcondorcet}
Bei Zhou and Klas Markstr{\"o}m.
\newblock An orderly algorithm for generation of {C}ondorcet domains, 2025.

\bibitem{zhou2025efficient}
Bei Zhou and S{\o}ren Riis.
\newblock An efficient heuristic search algorithm for discovering large
  condorcet domains.
\newblock {\em 4OR}, 23:193 -- 216, 2023.

\end{thebibliography}

%-----------------------------------------------------------
\section*{Appendix A}\label{appendix}
Pseudocode for the inductive search algorithm.

\begin{algorithm}[H]
\caption{Inductive extension from $n$ to $n{+}1$ alternatives}\label{alg:inductive}
\begin{algorithmic}[1]
\Require{A list $L_n$ of representatives of non-isomorphic large CDs on $[n]$, each represented by a complete assignment $\mathcal{N}_{[n]}$;\\
Allowed peak-pit rules $\mathcal{R}$ (typically $\{1N3,3N1,2N1,2N3\}$); size threshold $\beta$.}
\Ensure{A list $L_{n+1}$ of representatives of non-isomorphic candidate CDs on $[n{+}1]$ of size at least $\beta$.}
\State $L_{n+1}\gets \emptyset$
\State \textbf{/* Step 1: instantiate $n$-alternative assignments on each $n$-subset */} 
\For{$i\in[n{+}1]$}
  \State $I_i \gets \{\mathcal{N}([n{+}1]\setminus\{i\}) : \mathcal{N}\in L_n\}$ \Comment{order-preserving relabelling}
\EndFor
\State \textbf{/* Step 2: enumerate compatible overlaps and merge */} 
\For{ordered pairs $(i,j)\in[n{+}1]^2$ with $i\neq j$}
  \State $C \gets [n{+}1]\setminus\{i,j\}$
  \State Group $I_i$ and $I_j$ by their restriction to triples inside $C$
  \For{each common group key $K$}
    \For{each $\mathcal{N}^i\in I_i(K)$}
      \For{each $\mathcal{N}^j\in I_j(K)$}
        \State $\mathcal{N}_{\mathrm{partial}} \gets \mathcal{N}^i \cup \mathcal{N}^j$
        \State \textbf{/* Step 3: complete unassigned triples with pruning */} 
        \State Run backtracking on $\mathcal{U}(\mathcal{N}_{\mathrm{partial}})$, assigning only rules in $\mathcal{R}$
        \State \hspace{1.6em}and pruning partial assignments that cannot be extended to a copious domain
        \For{each completed assignment $\widehat{\mathcal{N}}$ found}
          \State Compute $|\mathcal{D}(\widehat{\mathcal{N}})|$
          \If{$|\mathcal{D}(\widehat{\mathcal{N}})|\ge \beta$}
            \State Add $\widehat{\mathcal{N}}$ to $L_{n+1}$
          \EndIf
        \EndFor
      \EndFor
    \EndFor
  \EndFor
\EndFor
\State Remove isomorphic copies from $L_{n+1}$
\State \Return $L_{n+1}$
\end{algorithmic}
\end{algorithm}

\end{document}